\documentclass[12pt]{article}    

\usepackage[margin=0.75in]{geometry} 
\usepackage{amsmath,amssymb,amsthm}

\newtheorem{theorem}{Theorem}[section]

\newtheorem{lemma}[theorem]{Lemma}
\newtheorem{corollary}[theorem]{Corollary}

\newtheorem{remark}{Remark}

\newcommand{\bt}{\beta}

\newcommand{\ka}{\kappa}

\newcommand{\s}{\sigma}

\newcommand{\be}{\begin{equation}}
\newcommand{\ee}{\end{equation}}
\newcommand{\bea}{\begin{eqnarray}}
\newcommand{\eea}{\end{eqnarray}}
\newcommand{\no}{\nonumber}

\numberwithin{equation}{section}
\begin{document}

\title{Hankel determinant and orthogonal polynomials arising from the matrix model in 2D quantum gravity}
\author{Chao Min\thanks{School of Mathematical Sciences, Huaqiao University, Quanzhou 362021, China; Email: chaomin@hqu.edu.cn}\: and Yadan Ding\thanks{School of Mathematical Sciences, Huaqiao University, Quanzhou 362021, China}}


\date{November 30, 2024}
\maketitle
\begin{abstract}
We study the Hankel determinant and orthogonal polynomials with respect to the two-parameter weight function
$$
w(x)=w(x;t_1, t_2):=\exp(-x^6-t_2 x^4-t_1 x^2),\qquad x\in\mathbb{R},
$$
with $t_1,\; t_2 \in \mathbb{R}$. This problem arises from the matrix model in 2D quantum gravity investigated by Fokas, Its and Kitaev [Commun. Math. Phys. \textbf{142} (1991) 313--344].
By making use of the ladder operator approach, we find that the recurrence coefficient $\beta_{n}(t_1,t_2)$ for the monic orthogonal polynomials satisfies a nonlinear fourth-order difference equation, which is within the discrete Painlev\'{e} I hierarchy. We show that the orthogonal polynomials satisfy a second-order linear differential equation whose coefficients are all expressed in terms of $\beta_{n}(t_1,t_2)$. The relations between the logarithmic partial derivative of the Hankel determinant, the nontrivial leading coefficient of the monic orthogonal polynomials, and the recurrence coefficient are established. By using Dyson's Coulomb fluid approach, we obtain the large $n$ asymptotic expansions of the recurrence coefficient $\beta_{n}(t_1,t_2)$, the nontrivial leading coefficient $\mathrm{p}(n,t_1,t_2)$, the normalized constant $h_n(t_1,t_2)$ and the Hankel determinant $D_{n}(t_1,t_2)$.
\end{abstract}

$\mathbf{Keywords}$: Hankel determinants; Orthogonal polynomials; Discrete Painlev\'{e} I hierarchy;

Ladder operators; Coulomb fluid; Large $n$ asymptotics.


$\mathbf{Mathematics\:\: Subject\:\: Classification\:\: 2020}$: 42C05, 33C45, 41A60.

\section{Introduction}
In the celebrated paper \cite{Fokas1991}, Fokas, Its and Kitaev presented an algorithmic method for deriving discrete analogues of Painlev\'{e} equations and for utilizing these equations to characterize similarity reductions of spatially discrete integrable evolution equations. It was shown in \cite{Fokas1991} that the algorithmic method's results are highly valuable for investigating the partition function of the matrix model in 2D quantum gravity associated with the measure $\exp(-t_3 z^6-t_2 z^4-t_1 z^2)$, which is equivalent to the following weight function by scaling:
\be\label{wf}
w(x)=w(x;t_1, t_2):=\exp(-x^6-t_2 x^4-t_1 x^2),\qquad x\in\mathbb{R},
\ee
with two parameters $t_1,\; t_2 \in \mathbb{R}$.

In this paper, we consider the Hankel determinant associated with the weight function (\ref{wf}), that is,
$$
D_{n}(t_1, t_2):=\det\big(\mu_{i+j}(t_1, t_2)\big)_{i,j=0}^{n-1},
$$
where $\mu_j(t_1, t_2)$ is the $j$th moment given by
$$
\mu_j(t_1, t_2):=\int_{-\infty}^{\infty}x^jw(x;t_1, t_2)dx,\qquad j=0,1,2,\dots.
$$

It is well-known that Hankel determinants are closely related to partition functions in random matrix theory. Indeed, let $Z_n(t_1,t_2)$ be the partition function for unitary random matrix ensemble associated with the weight function (\ref{wf}), i.e.,
\be\label{pf}
Z_n(t_1,t_2):=\int_{\mathbb{R}^n}\prod_{1\leq i<j\leq n}(x_i-x_j)^2\prod_{k=1}^n w(x_k;t_1, t_2)dx_k.
\ee
Then, there is only a constant difference between $D_{n}(t_1, t_2)$ and $Z_n(t_1,t_2)$ using Andr\'{e}ief's or Heine's integration formula \cite[(2.2.11)]{Szego}:
$$
D_n(t_1, t_2)=\frac{1}{n!}Z_n(t_1, t_2).
$$

It is also known that the Hankel determinant $D_n(t_1,t_2)$ can be expressed as the product of the normalized constants for the corresponding orthogonal polynomials (see, e.g., \cite[(2.1.6)]{Ismail}):
\be\label{hankel}
D_n(t_1,t_2)=\prod_{j=0}^{n-1}h_{j}(t_1,t_2),
\ee
where $h_{j}(t_1,t_2)>0$ is defined by
\be\label{pn}
h_{j}(t_1, t_2)\delta_{jk}=\int_{-\infty}^{\infty} P_{j}(x;t_1, t_2) P_{k}(x;t_1, t_2) w(x;t_1, t_2) d x,  \qquad j,k=0,1,2,\ldots.
\ee
Here $\delta_{jk}$ denotes the Kronecker delta, and $P_{j}(x;t_1, t_2)$ are the \textit{monic} polynomials of degree $j$, orthogonal with respect to the weight (\ref{wf}).

We mention that the orthogonal polynomials $P_{j}(x;t_1, t_2)$ are semi-classical orthogonal polynomials since the weight (\ref{wf}) satisfies the Pearson equation (see, e.g., \cite[Section 1.1.1]{VanAssche})
$$
\frac{d}{dx} \left(\s(x)w(x)\right)=\tau(x)w(x),
$$
with $\s(x)=x$ and $\tau(x)=1-2t_1 x^2-4t_2 x^4-6x^6$.

When $t_2=0$, the weight (\ref{wf}) is reduced to
\be\label{wf1}
w(x)=\exp(-x^6-t_1 x^2),\qquad x\in\mathbb{R},
\ee
with $t_1 \in \mathbb{R}$.
Orthogonal polynomials with respect to the weight (\ref{wf1}) have been considered in \cite{Magnus}, where a difference equation for the recurrence coefficient has been found. Some asymptotic properties for the Hankel determinant and orthogonal polynomials with the weight (\ref{wf1}) have been studied recently \cite{MWC}. Orthogonal polynomials with a slightly more general weight $w(x)=|x|^{\lambda}\exp(-x^6-t_1 x^2),\;x\in\mathbb{R},\;\lambda>-1,\;t_1 \in \mathbb{R}$ have been investigated in \cite{Clarkson2,WZC}.

Note that the weight (\ref{wf}) is an even function, which leads to the fact that
$
P_{n}(-x;t_1,t_2)=(-1)^nP_{n}(x;t_1,t_2)
$
and $P_{n}(x;t_1,t_2)$ has the monomial expansion \cite[p. 21]{Chihara1978}
\be\label{pn1}
P_{n}(x;t_1,t_2)=x^n+\mathrm{p}(n,t_1,t_2)x^{n-2}+\cdots,  \qquad n=0,1,2,\ldots,
\ee
where $\mathrm{p}(n,t_1,t_2)$ is the nontrivial leading coefficient of $P_{n}(x;t_1,t_2)$, and $\mathrm{p}(0,t_1,t_2)$=$\mathrm{p}(1,t_1,t_2)$=0. The three-term recurrence relation for the orthogonal polynomials reads
\be\label{sxdt}
xP_{n}(x;t_1,t_2)=P_{n+1}(x;t_1,t_2)+\beta_{n}(t_1,t_2)P_{n-1}(x;t_1,t_2),
\ee
supplemented by the initial conditions $P_{0}(x;t_1,t_2)=1,\; \beta_{0}(t_1,t_2)P_{-1}(x;t_1,t_2)=0$.

From (\ref{pn1}) and (\ref{sxdt}) we have
\be\label{beta1}
\beta_{n}(t_1,t_2)=\mathrm{p}(n,t_1,t_2)-\mathrm{p}(n+1,t_1,t_2),
\ee
and $\beta_{0}(t_1,t_2)=0$.
Taking a telescopic sum of (\ref{beta1}), we get an important identity
\be\label{beta3}
\sum_{j=0}^{n-1}\beta_{j}(t_1,t_2)=-\mathrm{p}(n,t_1,t_2).
\ee
We also have from (\ref{pn}) and (\ref{sxdt}) that
$$
\beta_{n}(t_1,t_2)=\frac{h_{n}(t_1,t_2)}{h_{n-1}(t_1,t_2)}.
$$

This paper is organized as follows. In Section 2, by applying the ladder operators and compatibility conditions to orthogonal polynomials with respect to the weight (\ref{wf}), we find that the recurrence coefficient $\beta_{n}(t_1,t_2)$ satisfies the second member of the discrete Painlev\'{e} I hierarchy. We also obtain the second-order differential equation satisfied by the orthogonal polynomials, and establish the relation between the nontrivial leading coefficient $\mathrm{p}(n,t_1,t_2)$ and the recurrence coefficient $\beta_{n}(t_1,t_2)$. In Section 3, we use Dyson's Coulomb fluid approach to study the large $n$ asymptotics of many quantities, including the recurrence coefficient $\beta_{n}(t_1,t_2)$, the nontrivial leading coefficient $\mathrm{p}(n,t_1,t_2)$, the normalized constant $h_n(t_1,t_2)$ and the Hankel determinant $D_{n}(t_1,t_2)$. Finally, we present some concluding remarks in Section 4.

\section{Ladder operators and compatibility conditions}
The ladder operator approach developed by Chen and Ismail \cite{Chen1997,ChenIsmail2005} is a highly useful and potent tool for analyzing the recurrence coefficients of orthogonal polynomials and the associated Hankel determinants; see e.g. \cite{ChenIts,Dai2010,Min2020}. Following \cite{Chen1997,ChenIsmail2005}, the orthogonal polynomials associated with the weight (\ref{wf}) satisfy the lowering and raising operators
\be\label{lo}
\left(\frac{d}{dx}+B_{n}(x)\right)P_{n}(x)=\beta_{n}A_{n}(x)P_{n-1}(x),
\ee
\be\label{ro}
\left(\frac{d}{dx}-B_{n}(x)-\mathrm{v}'(x)\right)P_{n-1}(x)=-A_{n-1}(x)P_{n}(x),
\ee
where $\mathrm{v}(x):=-\ln w(x)$ is the potential and
\be\label{and}
A_{n}(x):=\frac{1}{h_{n}}\int_{-\infty}^{\infty}\frac{\mathrm{v}'(x)-\mathrm{v}'(y)}{x-y}P_{n}^{2}(y)w(y)dy,
\ee
\be\label{bnd}
B_{n}(x):=\frac{1}{h_{n-1}}\int_{-\infty}^{\infty}\frac{\mathrm{v}'(x)-\mathrm{v}'(y)}{x-y}P_{n}(y)P_{n-1}(y)w(y)dy.
\ee
Note that we have suppressed the $t_1, t_2$ dependence of many quantities such as $\bt_n$ and $h_n$ for convenience.

Using the definitions of the functions $A_n(x)$ and $B_n(x)$ and with the aid of the three-term recurrence relation, it can be shown that
$A_n(x)$ and $B_n(x)$ satisfy the following compatibility conditions \cite{Chen1997,ChenIsmail2005}:
\be
B_{n+1}(x)+B_{n}(x)=x A_{n}(x)-\mathrm{v}'(x), \tag{$S_{1}$}
\ee
\be
1+x(B_{n+1}(x)-B_{n}(x))=\beta_{n+1}A_{n+1}(x)-\beta_{n}A_{n-1}(x). \tag{$S_{2}$}
\ee
It turns out that there is an equation which gives a better insight into the recurrence coefficient if ($S_{1}$) and ($S_{2}$) are suitably combined \cite{ChenIts}:
\be
B_{n}^{2}(x)+\mathrm{v}'(x)B_{n}(x)+\sum_{j=0}^{n-1}A_{j}(x)=\beta_{n}A_{n}(x)A_{n-1}(x). \tag{$S_{2}'$}
\ee

Recall that our weight function reads
$$
w(x)=\exp(-x^6-t_2 x^4-t_1 x^2),\qquad x\in\mathbb{R},
$$
with two parameters $t_1,\; t_2 \in \mathbb{R}$. We have
\be\label{vx}
\mathrm{v}(x)=-\ln w(x)=x^6+t_2 x^4+t_1 x^2
\ee
and
\be\label{v1}
\frac{\mathrm{v}'(x)-\mathrm{v}'(y)}{x-y}=6(x^4+x^3 y+x^2 y^2+x y^3+y^4)+4t_2(x^2+xy+y^2)+2t_1.
\ee
Substituting (\ref{v1}) into the definitions of $A_n (x)$ and $B_n (x)$ in (\ref{and}) and (\ref{bnd}), using the parity of the integrands and the three-term recurrence relation, we obtain
\be\label{an}
A_{n}(x)=6x^4+2x^2(2t_2+3\beta_{n}+3\beta_{n+1})+2t_1+4t_2 (\beta_{n}+\beta_{n+1})+6R_{n},
\ee
\be\label{bn}
B_{n}(x)=6 x^3\beta_{n}+2x(2t_2\beta_{n}+3r_{n}),
\ee
where $R_n$ and $r_n$ are the auxiliary quantities defined by
\be\label{Rnt}
R_n:=\frac{1}{h_n}\int_{-\infty}^{\infty}y^4 P_n^2(y)w(y)dy,
\ee
$$
r_n:=\frac{1}{h_{n-1}}\int_{-\infty}^{\infty}y^3 P_{n}(y)P_{n-1}(y)w(y)dy.
$$

Substituting the expressions of $A_{n}(x)$ and $B_{n}(x)$ in (\ref{an}) and (\ref{bn}) into ($S_1$), we have
\be\label{gs1}
R_n=r_n +r_{n+1}.
\ee
Substituting (\ref{an}) and (\ref{bn}) into ($S_{2}'$) and equating powers of $x$ on both sides, we obtain
\begin{align}
&x^6:\qquad r_{n}=\beta_{n}(\beta_{n-1}+\beta_{n}+\beta_{n+1}),\label{gs4}\\[8pt]
&x^4:\qquad n-2t_1 \beta_{n}+4t_2r_n+12\beta_{n} r_n=6\beta_{n}(R_n+R_{n-1})+8t_2 \beta_{n}(\beta_{n-1}+\beta_{n}+\beta_{n+1})\no\\
&\qquad \quad\:\:\: +6\beta_{n}(\beta_{n-1}+\beta_{n})(\beta_{n}+\beta_{n+1}),\label{gs5}\\[8pt]
&x^2:\qquad 2nt_2+6t_1 r_n+18r_n^2+24t_2\beta_{n}r_n+3\sum_{j=0}^{n-1}(\beta_{j}+\beta_{j+1})=4t_1 t_2\beta_{n}+6t_1\beta_{n}^2\no\\
&\qquad \quad\:\:\: +12t_2 \beta_{n}(R_n+R_{n-1})+18\beta_{n}(\beta_{n-1}+\beta_{n})R_{n}+18\beta_{n}(\beta_{n}+\beta_{n+1})R_{n-1}\no\\
&\qquad \quad\:\:\: +2(3t_1+4t_2^2)\beta_{n}(\beta_{n-1}+\beta_{n}+\beta_{n+1})+24t_2 \beta_{n}(\beta_{n-1}+\beta_{n})(\beta_{n}+\beta_{n+1}),\label{gs6}\\[8pt]
&x^0:\qquad nt_1 +3\sum_{j=0}^{n-1}R_{j}+2t_2 \sum_{j=0}^{n-1}(\beta_{j}+\beta_{j+1})=2\bt_n\big[t_1+3R_{n-1}+2t_2(\bt_{n-1}+\bt_n)\big]\no\\
&\qquad \quad\:\:\: \times\big[t_1+3R_{n}+2t_2(\bt_{n}+\bt_{n+1})\big]\label{gs7}.
\end{align}
The identities (\ref{gs1})--(\ref{gs7}) are very important and all will be used in the following analysis.
\begin{theorem}\label{th1}
The recurrence coefficient $\bt_n$ satisfies the fourth-order difference equation
\begin{align}\label{bcf}
&\:6\bt_n(\bt_{n-2}\bt_{n-1}+\bt_{n-1}^2+2\bt_{n-1}\bt_{n}+\bt_{n-1}\bt_{n+1}+\bt_{n}^2+2\bt_{n}\bt_{n+1}+\bt_{n+1}^2+\bt_{n+1}\bt_{n+2})\no\\
&+4t_2\bt_{n}(\bt_{n-1}+\bt_{n}+\bt_{n+1})+2t_1 \bt_{n}=n,
\end{align}
which is the second member of the discrete Painlev\'{e} I hierarchy \cite[(2.9)]{Cresswell1999} with parameters
$c_0=0,\; c_1=0,\; c_2=-2t_1,\; c_3=-4t_2,\; c_4=-6$.
\end{theorem}
\begin{proof}
From (\ref{gs1}) and (\ref{gs4}) we have
\be\label{Rn1}
R_n=\beta_{n-1}\beta_{n}+(\bt_n+\bt_{n+1})^2+\bt_{n+1}\bt_{n+2}.
\ee
Substituting (\ref{gs4}) and (\ref{Rn1}) into (\ref{gs5}), we obtain (\ref{bcf}).
\end{proof}
\begin{remark}
The relation of this problem to the discrete Painlev\'{e} equations was first observed in \cite{Fokas1991}. The fourth-order equation in the discrete Painlev\'{e} I hierarchy \cite[(2.9)]{Cresswell1999} has also appeared in other (one-parameter) sextic Freud weight problems \cite{Clarkson2,DM,Magnus,WZC}.
\end{remark}
In the next theorem, we show that our orthogonal polynomials satisfy a second-order ordinary differential equation, with the coefficients all expressed in terms of the recurrence coefficient $\bt_n$.
\begin{theorem}
The orthogonal polynomials $P_n(x)$ satisfy the second-order differential equation
\begin{align}\label{ode}
&P_n''(x)-\bigg(\mathrm{v}'(x)+\frac{A_n'(x)}{A_n(x)}\bigg)P_n'(x)+\bigg(B_n'(x)-B_n^2(x)-\mathrm{v}'(x)B_n(x)\no\\[8pt]
&+\beta_nA_n(x)A_{n-1}(x)-\frac{A_n'(x)B_n(x)}{A_n(x)}\bigg)P_n(x)=0,
\end{align}
where $\mathrm{v}(x)$ is given by (\ref{vx}), and $A_{n}(x)$ and $B_n(x)$ are expressed in terms of $\bt_n$ as follows:
\begin{align}\label{an1}
A_{n}(x)=&\:6x^4+2x^2(2t_2+3\beta_{n}+3\beta_{n+1})+2t_1+4t_2 (\beta_{n}+\beta_{n+1})\no\\
&+6\left[\beta_{n-1}\beta_{n}+(\bt_n+\bt_{n+1})^2+\bt_{n+1}\bt_{n+2}\right],
\end{align}
\be\label{bn1}
B_{n}(x)=6 x^3\beta_{n}+2x\beta_{n}(2t_2+3\beta_{n-1}+3\beta_{n}+3\beta_{n+1}).
\ee
\end{theorem}
\begin{proof}
Eliminating $P_{n-1}(x)$ from the ladder operator equations (\ref{lo}) and (\ref{ro}), we obtain (\ref{ode}). The expressions of $A_n(x)$ and $B_n(x)$ in (\ref{an1}) and (\ref{bn1}) come from (\ref{an}) and (\ref{bn}), with the aid of (\ref{gs4}) and (\ref{Rn1}).
\end{proof}
At the end of this section, we present the relation between the nontrivial leading coefficient $\mathrm{p}(n,t_1,t_2)$ and the recurrence coefficient $\beta_n$, which will be used to derive the large $n$ asymptotics of $\mathrm{p}(n,t_1,t_2)$ in the next section.
\begin{theorem}\label{th2}
The nontrivial leading coefficient $\mathrm{p}(n,t_1,t_2)$ is expressed in terms of the recurrence coefficient $\beta_n$ as follows:
\be\label{pn11}
\mathrm{p}(n,t_1,t_2)=\frac{1}{2} \bt_{n}\big[1-n-4t_2\bt_{n-1}\bt_{n+1}-6\bt_{n-1}\bt_{n+1}(\bt_{n-2}+\bt_{n-1}+\bt_{n}+\bt_{n+1}+\bt_{n+2})\big].
\ee
\end{theorem}
\begin{proof}
Taking account of (\ref{beta3}) we have
\be\label{beta4}
\sum_{j=0}^{n-1}(\bt_{j}+\bt_{j+1})=\bt_{n}-2\mathrm{p}(n,t_1,t_2).
\ee
Substituting (\ref{beta4}) into (\ref{gs6}) and eliminating $r_n$ and $R_{n}$ by (\ref{gs4}) and (\ref{Rn1}), we get the expression of $\mathrm{p}(n,t_1,t_2)$ in terms of $\bt_n$. By making use of (\ref{bcf}) to simplify the result, we arrive at (\ref{pn11}).
\end{proof}

\section{Large $n$ asymptotics}
Based on Dyson's Coulomb fluid approach \cite{Dyson}, for sufficiently large $n$, the eigenvalues of $n\times n$ Hermitian matrices of a unitary random matrix ensemble can be approximated as a continuous fluid with an equilibrium density $\s(x)$ supported on $J$ (a subset of $\mathbb{R}$). For the unitary ensemble whose partition function given by (\ref{pf}), the potential $\mathrm{v}(x)=x^6+t_2 x^4+t_1 x^2, \;x\in \mathbb{R}$ is even. In this case, when $x\mathrm{v}'(x)$ is positive and increasing on $\mathbb{R}^{+}$, i.e., $t_1\geq0,\; t_2\geq0$ (or $t_1 \geq \frac{4}{9}t_2^2, \; t_2 <0$), $J$ is a single symmetric interval, say, $(-b, b)$; see \cite[p. 203]{Saff}.
According to the work of Chen and Ismail \cite{ChenIsmail} (see also \cite{Saff}), the equilibrium density $\sigma(x)$ is determined by minimizing the free energy functional
$$
F[\s]:=\int_{-b}^{b}\s(x)\mathrm{v}(x)dx-\int_{-b}^{b}\int_{-b}^{b}\s(x)\ln|x-y|\s(y)dxdy,
$$
subject to the normalization condition
\be\label{con}
\int_{-b}^{b}\s(x)dx=n.
\ee
It is found that the density $\s(x)$ satisfies the integral equation
\be\label{ie}
\mathrm{v}(x)-2\int_{-b}^{b}\ln|x-y|\s(y)dy=A,\qquad x\in (-b, b),
\ee
where $A$ is the Lagrange multiplier for the constraint (\ref{con}). Note that $A$ is a constant independent of $x$ but it depends on $n$ and the parameters $t_1, t_2$.

Equation (\ref{ie}) is transformed into the following singular integral equation by taking a derivative with respect to $x$:
\be\label{sie}
\mathrm{v}'(x)-2P\int_{-b}^{b}\frac{\sigma(y)}{x-y}dy=0,\qquad x\in (-b, b),
\ee
where $P$ denotes the Cauchy principal value. Multiplying by $\frac{x}{\sqrt{b^2-x^2}}$ on both sides of (\ref{sie}) and integrating from $-b$ to $b$ with respect to $x$, we obtain
\be\label{sup2}
\int_{-b}^{b}\frac{x\:\mathrm{v}'(x)}{\sqrt{b^2-x^2}}dx=2\pi n,
\ee
where use has been made of (\ref{con}) and the important integral formula
$$
P\int_{-b}^{b}\frac{1}{(x-y)\sqrt{b^2-x^2}}dx=0,\qquad y\in (-b,b).
$$
The endpoint of the support of the equilibrium density, $b$, is determined by equation (\ref{sup2}).
Furthermore, it was shown in \cite{ChenIsmail} that as $n\rightarrow\infty$,
\begin{equation}\label{beta}
\beta_n=\frac{b^2}{4}\left(1+O\left(\frac{\partial^3 A}{\partial n^3}\right)\right).
\end{equation}

Substituting (\ref{vx}) into (\ref{sup2}), we obtain a sextic equation satisfied by $b$:
$$
b^2  \left(15b^4+12t_2b^2+8t_1\right)=16 n.
$$
Let $u=b^2$, then we get a cubic equation
$$
15 u^3+12t_2u^2+8t_1 u-16n=0,
$$
which has a unique real solution
\be\label{u}
u=\frac{2}{15}\left(\phi^{1/3}-2 t_2+\frac{4 t_2^2- 10 t_1}{ \phi^{1/3}}\right),
\ee
where
$$
\phi=225 n-8 t_2^3+30 t_1 t_2 +5 \sqrt{2025 n^2-36 n (4 t_2^3-15 t_1 t_2)+40 t_1^3-12 t_1^2 t_2^2}.
$$
It follows that as $n\rightarrow\infty$,
\begin{align}\label{mu11}
\frac{u}{4}=&\frac{n^{1/3}}{\kappa}-\frac{t_2}{15}-\frac{2(5t_1-2t_2^2)}{15\ka^2n^{1/3}}+\frac{2t_2(15t_1-4t_2^2)}{675\ka n^{2/3}}+ \frac{4 t_2(5t_1 -2t_2^2)(15t_1 -4t_2^2)}{10125\ka^2 n^{4/3}}\no\\[8pt]
&+\frac{2(375 t_1^3-900 t_1^2 t_2^2+420 t_1 t_2^4-56 t_2^6)}{455625\ka n^{5/3}} +O(n^{-7/3}),
\end{align}
where $\ka=\sqrt[3]{60}$, and this notation will be used throughout the following text.

From (\ref{ie}) and similarly as in \cite[Lemma 3]{Min2021}, we find the Lagrange multiplier $A$ is
$$
A=\frac{1}{\pi}\int_{-b}^{b}\frac{\mathrm{v}(x)}{\sqrt{b^2-x^2}}dx-2n\ln \frac{b}{2}=\frac{1}{16}u (5u^2+6t_2u+8t_1)-n\ln \frac{u}{4}.
$$
Substituting (\ref{u}) into the above and taking a large $n$ limit, we have
\begin{align}\label{aa}
A=&-\frac{1}{3}n\ln n+\left(\frac{1}{3}+\ln \ka\right)n+\frac{6t_2 n^{2/3}}{\ka^2}+\frac{2(5t_1-t_2^2)n^{1/3}}{5\ka}-\frac{2 t_2(5t_1-t_2^2)}{75}\no\\[8pt]
&-\frac{2(15t_1^2-12t_1 t_2^2+2t_2^4)}{45\ka^2 n^{1/3}}+\frac{2 t_2(375t_1^2-200t_1 t_2^2+28t_2^4)}{16875\ka n^{2/3}}+O(n^{-4/3}).
\end{align}

We are now ready to derive the large $n$ asymptotics of the recurrence coefficient $\bt_n(t_1, t_2)$ for fixed $t_1, t_2$.
\begin{theorem}\label{thm}
The recurrence coefficient $\bt_n$ has the following asymptotic expansion as $n\rightarrow\infty $:
\begin{align}\label{bnz}
\bt_n=&\frac{n^{1/3}}{\kappa}-\frac{t_2}{15}-\frac{2(5t_1-2t_2^2)}{15\ka^2n^{1/3}}+\frac{2t_2(15t_1-4t_2^2)}{675\ka n^{2/3}}+ \frac{4t_2(5t_1 -2t_2^2)(15t_1-4t_2^2)}{10125\ka^2 n^{4/3}}\no\\[8pt]
&+\frac{1500t_1^3-3600t_1^2 t_2^2 +1680t_1 t_2^4-224t_2^6+ 50625}{911250\ka n^{5/3}}-\frac{t_2}{270 n^{2}}+O(n^{-7/3}),
\end{align}
where $\ka=\sqrt[3]{60}$.
\end{theorem}
\begin{proof}
From (\ref{beta}), (\ref{mu11}) and (\ref{aa}), we find that $\bt_n$ has the large $n$ expansion of the form
\be\label{bt1}
\bt_n=a_{-1}n^{1/3}+a_0+\sum_{j=1}^{\infty}\frac{a_j}{n^{j/3}},
\ee
where $a_j,\; j=-1, 0, 1, \ldots$ are the expansion coefficients to be determined.
Substituting (\ref{bt1}) into the difference equation (\ref{bcf}), taking a large $n$ limit and equating powers of $n$ on both sides, we obtain the expansion coefficients $a_j,\; j=-1, 0, 1, \ldots$ one by one. The first few terms are
\begin{align}
&a_{-1}=\frac{1}{\kappa},\qquad\qquad\qquad\qquad a_{0}=-\frac{t_2}{15}, \qquad\qquad\qquad\qquad a_1=-\frac{2(5t_1-2t_2^2)}{15\ka^2},\nonumber\\[8pt]
&a_2=\frac{2t_2(15t_1-4t_2^2)}{675\ka },\qquad\;\;\; a_3=0,\qquad\qquad\;\; a_4=\frac{4t_2(5t_1 -2t_2^2)(15t_1-4t_2^2)}{10125\ka^2},\nonumber\\[8pt]
&a_5=\frac{1500t_1^3-3600t_1^2 t_2^2 +1680t_1 t_2^4-224t_2^6+ 50625}{911250\ka},\qquad\qquad\qquad a_6=-\frac{t_2}{270}.\no
\end{align}
The theorem is then established.
\end{proof}
Based on the above theorem, the large $n$ asymptotics of the nontrivial leading coefficient $\mathrm{p}(n,t_1,t_2)$ and the Hankel determinant $D_{n}(t_1,t_2)$ will be derived from their relations to the recurrence coefficient.
\begin{theorem}
The nontrivial leading coefficient of the monic orthogonal polynomials, $\mathrm{p}(n,t_1,t_2)$, has the following large $n$ asymptotic expansion:
\begin{align}\label{pnt}
\mathrm{p}(n,t_1,t_2)=&-\frac{3 n^{4/3}}{4 \ka}+\frac{t_2 n }{15}+\frac{ (5 t_1-2 t_2^2)n^{2/3}}{5\ka^2}+\frac{ (16 t_2^3-60 t_1 t_2+225)n^{1/3}}{450\ka}\no\\[8pt]
&-\frac{25 t_1^2-30 t_1 t_2^2+75 t_2+6 t_2^4}{2250}+\frac{ (5 t_1-2 t_2^2) (60 t_1 t_2-16 t_2^3-225) }{ 3375 \ka^2n^{1/3}}\no\\[8pt]
&+\frac{750 t_1^3-1800 t_1^2 t_2^2+6750 t_1 t_2+840 t_1 t_2^4-1800 t_2^3-112 t_2^6+16875 }{303750\ka n^{2/3}}\no\\[8pt]
&-\frac{t_2}{270 n}+O(n^{-4/3}).
\end{align}
\end{theorem}
\begin{proof}
Substituting (\ref{bnz}) into the expression of $\mathrm{p}(n,t_1,t_2)$ in (\ref{pn11}) and taking a large $n$ limit, we obtain the desired result.
\end{proof}
\begin{remark}
For consistency check, substituting (\ref{pnt}) into (\ref{beta1}) and taking a large $n$ limit, we obtain the large $n$ asymptotic expansion of $\bt_n$, which agrees precisely with the result in Theorem \ref{thm}.
\end{remark}
\begin{lemma}
The quantity $\frac{\partial}{\partial t_2} \ln {D}_n(t_1, t_2)$ is expressed in terms of the recurrence coefficient $\bt_n$ as follows:
\begin{align}\label{rb}
\frac{\partial}{\partial t_2} \ln {D}_n(t_1, t_2)=&\:2t_1 \bt_{n-1}\bt_{n}\bt_{n+1}-2t_1 \bt_{n}^{2}(\bt_{n-1}+\bt_{n}+\bt_{n+1})-4t_2 \bt_{n}^{2}(\bt_{n-1}+\bt_{n}+\bt_{n+1})^{2}\no\\
&-6\bt_{n}\big[\bt_{n-1}(\bt_{n-2}+\bt_{n-1}+\bt_{n})+\bt_{n}(\bt_{n-1}+\bt_{n}+\bt_{n+1})\big]\no\\
&\times\big[\bt_{n}(\bt_{n-1}+\bt_{n}+\bt_{n+1})+ \bt_{n+1}(\bt_{n}+\bt_{n+1}+\bt_{n+2})\big].
\end{align}
\end{lemma}
\begin{proof}
From (\ref{pn}) we have
$$
h_{n}(t_1,t_2)=\int_{-\infty}^{\infty}P_{n}^{2}(x;t_1,t_2)w(x;t_1,t_2)dx.
$$
Taking a partial derivative with respect to $t_2$ gives
$$
\frac{\partial}{\partial t_2} \ln h_{n}(t_1,t_2)=-R_n,
$$
where $R_n$ is given by (\ref{Rnt}). It follows from (\ref{hankel}) that
\be\label{pd}
\frac{\partial}{\partial t_2} \ln {D}_n(t_1, t_2)=- \sum_{j=0}^{n-1}R_j.
\ee
Eliminating $\sum_{j=0}^{n-1}(\beta_{j}+\beta_{j+1})$ from the combination of (\ref{gs6}) and (\ref{gs7}) and making use of (\ref{gs4}) and (\ref{Rn1}), we can express $\sum_{j=0}^{n-1}R_j$ in terms of the recurrence coefficient $\bt_n$. Using (\ref{bcf}) to simplify the result, we arrive at (\ref{rb}) from (\ref{pd}).
\end{proof}
\begin{theorem}
The Hankel determinant $D_n(t_1, t_2)$ has the large $n$ asymptotic expansion
\begin{align}\label{dnta}
\ln D_n(t_1, t_2)=&\:\frac{1}{6} n^2 \ln n-\frac{1}{4}(1+2\ln \ka)n^2-\frac{18t_2 n^{5/3}}{5 \ka^2}-\frac{3(5t_1-t_2^2) n^{4/3}}{10 \ka}+\left[\frac{2t_2(5t_1-t_2^2)}{75}+\ln (2\pi)\right]n\no\\[8pt]
&+\frac{(15t_1^2-12t_1t_2^2+2t_2^4)n^{2/3}}{15 \ka^2}-\frac{2t_2 (375t_1^2-200t_1t_2^2+28t_2^4)n^{1/3}}{5625 \ka}-\frac{\ln n}{12}\no\\[8pt]
&+C_0+\frac{C_1}{n^{1/3}}+\frac{C_2}{n^{2/3}}+\frac{C_3}{n}+O(n^{-4/3}),
\end{align}
where the constant term $C_0$ and the coefficients of the first few higher order terms are
\begin{align}
&C_0=\frac{t_2^2 (125 t_1^2 -50 t_1 t_2^2+ 6 t_2^4)}{9375}-\frac{t_1^3}{135}-\frac{\ln 3}{12}+\zeta'(-1),\no\\[8pt]
&C_1=\frac{ t_2 (21000 t_1^3-21000 t_1^2 t_2^2+6720 t_1 t_2^4-704 t_2^6+118125)}{354375 \ka^2},\no\\[8pt]
&C_2=\frac{1875 t_1^4-6000 t_1^3 t_2^2+4200 t_1^2 t_2^4+168750 t_1-1120 t_1 t_2^6-50625 t_2^2+ 104 t_2^8}{1518750\ka },\no\\[8pt]
&C_3=-\frac{t_2 (15t_1 -4t_2^2)}{2025},\no
\end{align}
and $\ka=\sqrt[3]{60},\; \zeta'(\cdot)$ is the derivative of the Riemann zeta function.
\end{theorem}
\begin{proof}
Substituting (\ref{bnz}) into (\ref{rb}) and taking a large $n$ limit, we find
\begin{align}
\frac{\partial}{\partial t_2} \ln {D}_n(t_1, t_2)=&-\frac{18 n^{5/3}}{5\ka^2}+\frac{3t_2 n^{4/3}}{5\ka}+\frac{2(5t_1-3t_2^2)n}{75}-\frac{8t_2(3t_1-t_2^2)n^{2/3}}{15\ka^2}\no\\[8pt]
&-\frac{2(75t_1^2-120t_1 t_2^2+28t_2^4)n^{1/3}}{1125\ka}+\frac{2t_2(125t_1^2 -100t_1 t_2^2+18t_2^4)}{9375}\no\\[8pt]
&+\frac{3000t_1^3-9000t_1^2 t_2^2+4800 t_1 t_2^4-704t_2^6+ 16875}{50625 \ka^2 n^{1/3}}\no\\[8pt]
&-\frac{t_2(6000t_1^3-8400t_1^2 t_2^2+3360t_1 t_2^4-416t_2^6+50625)}{759375 \ka n^{2/3}}\no\\[8pt]
&-\frac{5t_1-4t_2^2}{675n}+O(n^{-4/3}).\no
\end{align}
Integrating the above with respect to $t_2$ over the interval $[0, t_2]$ gives
\begin{align}\label{dnt}
\frac{\ln D_n(t_1, t_2)}{\ln D_n(t_1, 0)}=&-\frac{18t_2 n^{5/3}}{5 \ka^2}+\frac{3t_2^2 n^{4/3}}{10 \ka}+\frac{2t_2(5t_1-t_2^2)n}{75}-\frac{2t_2^2(6t_1-t_2^2)n^{2/3}}{15 \ka^2}\no\\[8pt]
&-\frac{2t_2 (375t_1^2-200t_1t_2^2+28t_2^4)n^{1/3}}{5625 \ka} +\frac{t_2^2 (125 t_1^2 -50 t_1 t_2^2+ 6 t_2^4)}{9375}\no\\[8pt]
&+\frac{ t_2 (21000 t_1^3-21000 t_1^2 t_2^2+6720 t_1 t_2^4-704 t_2^6+118125)}{354375 \ka^2n^{1/3}}\no\\[8pt]
&-\frac{ t_2^2 (6000 t_1^3-4200 t_1^2 t_2^2+1120 t_1 t_2^4 -104 t_2^6+50625)}{1518750\ka n^{2/3}}\no\\[8pt]
&-\frac{t_2 (15t_1 -4t_2^2)}{2025 n}+O(n^{-4/3}).
\end{align}
The large $n$ asymptotic expansion of $\ln D_n(t_1, 0)$ has recently been obtained in \cite[Theorem 3.4]{MWC} ($t$ replaced by $-t_1$ for our problem):
\begin{align}\label{dnt1}
\ln D_n(t_1,0)=&\:\frac{1}{6} n^2 \ln n-\left(\frac{1}{4}+\frac{\ln \ka}{2}\right)n^2-\frac{3t_1 n^{4/3}}{2 \ka} +n \ln (2\pi)+\frac{t_1^2 n^{2/3}}{\ka^2} -\frac{\ln n}{12} -\frac{t_1^3}{135}\no\\[8pt]
&+\zeta'(-1)-\frac{\ln 3}{12}+\frac{t_1 (t_1^3+90)}{810 \ka n^{2/3}}+\frac{2t_1^5}{6075\ka^2 n^{4/3}} +O(n^{-2}).
\end{align}
The sum of (\ref{dnt}) and (\ref{dnt1}) gives the result in (\ref{dnta}).
\end{proof}

\begin{corollary}
The normalized constant $h_n(t_1,t_2)$ has the large $n$ asymptotic expansion
\begin{align}
\ln h_n(t_1,t_2)=&\:\frac{1}{3}n\ln n-\left(\frac{1}{3}+\ln\ka\right)n -\frac{6t_2 n^{2/3}}{\ka^2}-\frac{2(5t_1-t_2^2)n^{1/3}}{5 \ka}+\frac{1}{6}\ln n\no\\[8pt]
&+\ln (2 \pi)-\frac{1}{2}\ln \ka+\frac{2 t_2(5t_1-t_2^2)}{75}+\frac{2(15t_1^2-45t_2-12t_1 t_2^2+2t_2^4)}{45 \ka^2 n^{1/3}}\no\\[8pt]
&-\frac{750 t_1^2 t_2-400 t_1 t_2^3+5625 t_1-1125 t_2^2+ 56 t_2^5}{16875 \ka n^{2/3}}-\frac{1}{36n}+O(n^{-4/3}).\no
\end{align}
\end{corollary}
\begin{proof}
From (\ref{hankel}) we have
$$
\ln h_n(t_1,t_2)=\ln D_{n+1}(t_1, t_2)-\ln D_n(t_1, t_2).
$$
Substituting (\ref{dnta}) into above and taking a large $n$ limit, we obtain the desired result.
\end{proof}
In the end, we would like to point out that all the large $n$ asymptotic expansions obtained in this section can be easily extended up to any higher order.

\section{Discussion}
In the past few decades, most of the studies on semi-classical orthogonal polynomials were about the one-parameter (time parameter) weight functions.
In this paper, we investigated the Hankel determinant and orthogonal polynomials associated with a two-parameter Freud weight. It can be seen that the large $n$ asymptotics for this problem is derived based on the corresponding one-parameter problem ($t_2=0$). The study helps us to better understand the Hankel determinants and orthogonal polynomials with two or more parameter weight functions.

\section*{Acknowledgments}
This work was partially supported by the National Natural Science Foundation of China under grant number 12001212, by the Fundamental Research Funds for the Central Universities under grant number ZQN-902 and by the Scientific Research Funds of Huaqiao University under grant number 17BS402.

\section*{Conflict of Interest}
The authors have no competing interests to declare that are relevant to the content of this article.
\section*{Data Availability Statements}
Data sharing not applicable to this article as no datasets were generated or analysed during the current study.

\end{document}